\documentclass[12pt]{article}
\usepackage{amsmath}
\usepackage{amssymb}
\usepackage[a4paper]{geometry}
\usepackage{amsthm}
\usepackage{csquotes}
\usepackage{subcaption}
\usepackage{graphicx}
\usepackage{tikz}
\usepackage{tikz-cd}
\usepackage{mathrsfs}
\usepackage{bm}
\usepackage{bbold}
\usepackage{lscape}
\usepackage{slashed}
\usepackage[hidelinks]{hyperref}
\usepackage{listings}
\usepackage{verbatim}
\usepackage{orcidlink}

\usepackage{soul, color}

\hypersetup{
    colorlinks = true,
    linkcolor = {blue},
    citecolor = {blue},
    anchorcolor = {green},
    citecolor = {red},
    filecolor = {cyan},
    menucolor = {red},
    runcolor = {cyan},
    urlcolor = {magenta}
}
\newcommand{\twidth}{6in}
\textwidth=\twidth
\newcommand{\su}{{\mathfrak{su}}}

\renewcommand{\O}{{\mathrm{O}}}
\newcommand{\SU}{{\mathrm{SU}}}

\newcommand{\R}{{\mathbb{R}}}

\newcommand{\Z}{{\mathbb{Z}}}

\newcommand{\bi}{{\bm{i}}}
\newcommand{\bj}{{\bm{j}}}
\newcommand{\bk}{{\bm{k}}}

\newcommand{\beq}{\begin{equation}}
\newcommand{\eeq}{\end{equation}}
\newcommand{\bea}{\begin{eqnarray}}
\newcommand{\eea}{\end{eqnarray}}
\newcommand{\bal}{\begin{align}}
\newcommand{\eal}{\end{align}}
\newcommand{\bml}{\begin{multline}}
\newcommand{\eml}{\end{multline}}

\newcommand{\wh}{\widehat}
\newcommand{\lto}{\longrightarrow}

\def \d{\mathrm{d}}
\newcommand{\ol}{\overline}
\newcommand{\tr}{{\rm tr}\, }

\newcommand{\id}{{\rm Id}}

\newcommand{\ip}[1]{\langle#1\rangle}

\newcommand{\abs}[1]{\left\lvert#1\right\rvert}

\makeatletter
\newcommand\xleftrightarrow[2][]{%
  \ext@arrow 9999{\longleftrightarrowfill@}{#1}{#2}}
\newcommand\longleftrightarrowfill@{%
  \arrowfill@\leftarrow\relbar\rightarrow}
\makeatother

\newtheorem{theorem}{Theorem}

\newtheorem{example}[theorem]{Example}

\newtheorem{lemma}[theorem]{Lemma}

\newtheorem{proposition}[theorem]{Proposition}
\newtheorem{remark}[theorem]{Remark}

\begin{document}
\renewcommand*{\thefootnote}{\fnsymbol{footnote}}
\begin{titlepage}
\begin{center}
{\Large
Locations of JNR skyrmions \par}


\vspace{10mm}
{\Large Josh Cork\footnote{Email address: \texttt{josh.cork@leicester.ac.uk}}$^{\rm 1}$\,\orcidlink{0000-0002-9006-0108} and Linden Disney--Hogg\footnote{Email address: \texttt{a.l.disney-hogg@leeds.ac.uk}.}$^{\rm 2}$\,\orcidlink{0000-0002-6597-2463}}\\[10mm]

\noindent {\em ${}^{\rm 1}$ School of Computing and Mathematical Sciences\\
University of Leicester, University Road, Leicester, United Kingdom
}\\
\smallskip
\noindent {\em ${}^{\rm 2}$ School of Mathematics, University of Leeds\\ Woodhouse Lane, Leeds, United Kingdom }\\[10mm]
{\Large \today}
\vspace{15mm}
\begin{abstract}
    We develop and prove new geometric and algebraic characterisations for locations of constituent skyrmions, as well as their signed multiplicity, using Sutcliffe's JNR ansatz. Some low charge examples, and their similarity to BPS monopoles, are discussed. In addition, we provide Julia code for the further numerical study and visualisation of JNR skyrmions. 
\end{abstract}
\end{center}
\end{titlepage}
\renewcommand*{\thefootnote}{\arabic{footnote}}
\section{Introduction}
Skyrmions are topological solitons in a low-energy model of nuclei \cite{skyrme1962nucl,manton2022skyrmions}. They are classified by a topological charge $N$: an integer physically interpreted as the baryon number. On Euclidean $\R^3$, the Skyrme field equations do not admit explicit exact solutions, and so much work has gone into finding frameworks for approximating skyrmions. These toy models explain qualitative behaviour and are more amenable to analytic techniques which yield insight into baryon structure and allow a quantum treatment. In recent work of Sutcliffe \cite{sutcliffe2024jnr}, a simple model of Skyrme fields arising from Jackiw--Nohl--Rebbi (JNR) instantons has been found. In this paper we seek to illuminate one key aspect of these approximate skyrmions: their constituent locations. 

Our investigation yields two key results. 
\begin{enumerate}
    \item[Theorem \ref{thm:signs-locations-via-f}.] We identify the locations with the critical points of a Morse function, and show that the nature of the critical point determines the signed multiplicity of the location.
    \item[Theorem \ref{thm: 2-skymrion locations are foci}.] For JNR skyrmions with poles in a plane in $\R^3$ we give a geometric description of the locations, in particular seeing that in the $N=2$ case these are foci arising from the famous Poncelet porism.
\end{enumerate}

In addition, through numerical approaches (for which we provide Julia code) we classify the behaviour of the locations of skyrmions that form the 1-parameter twisted line scattering of Walet \cite{Walet:1996he}. This latter development is important because of its connection with the twisted line scattering of monopoles \cite{HoughtonSutcliffe1996monopole}. An unexplained analogy between Euclidean $\SU(2)$ Bogomol'nyi--Prasad--Sommerfield (BPS) monopoles and skyrmions has long been observed, for the most part due to the success of the rational map approximation of skyrmions \cite{HoughtonMantonSutcliffe1998rational}. However the rational map approximation fails to fully capture the analogy between the two twisted line scatterings. Here we demonstrate how the JNR approximation, despite being slightly coarser than rational maps in terms of energy, still captures the intricate behaviour of location creation and annihilation. This suggests that future explanation of the analogy could stem from studying the JNR approximation, or instantons more generally.

This paper is organised as follows. In \S\ref{sec: skyrme fields} we lay out the preliminary definitions of what a skyrmion is and the corresponding JNR approximation of Sutcliffe. In \S\ref{sec: locations} we define what we mean by constituent locations, as well as the associated sign. There we shall also show how to calculate these locations and signs in terms of certain critical points. \S\ref{sec: examples} gives examples of these locations and their signs for baryon numbers $N=1,2,3$, including the twisted line scattering. \S\ref{sec:concl} discusses some further outlook and open problems.

\section{Skyrme fields and the JNR ansatz}\label{sec: skyrme fields}

A Skyrme field is a smooth map $U:\R^3\lto\SU(2)$ satisfying the space-compactifying boundary condition $U(x)\to1$ as $|x|\to\infty$. This boundary condition identifies $U$ with a map $S^3\lto S^3$ which has a topological charge, the degree $N\in\Z=\pi_3(S^3)$, physically identified with the baryon number. A skyrmion is a Skyrme field which minimises the Skyrme energy
\begin{align}\label{sky-energy}
    E=\int(|L|^2+|L\wedge L|^2)\,\d^3x,
\end{align}
where $L=U^{-1}\d U$, and $|\xi|^2\,\d^3x:=\tfrac{1}{2}\tr(\xi\wedge\star\xi^\dagger)$ for $\su(2)$-valued\footnote{Later we identify $\SU(2)$ with the unit quaternions, in which case this formula is the same but with $\tfrac{1}{2}\tr$ replaced with the real part $\Re$, and hermitian conjugate replaced with quaternionic conjugate.} forms. No analytic solutions are known, but many locally minimal energy fields may be found numerically \cite{battye2002skyrmions,gudnasonhalcrow2022smorgaasbord}. As such, to complement numerics and allow for a (semi-classical) quantum treatment, approximate descriptions are required. We remark that although the term \textit{skyrmion} is typically reserved for the energy minimisers, we shall often abuse this terminology to refer to generic Skyrme field configurations as well.

One of the most powerful approximations of skyrmions is due to Atiyah--Manton \cite{AtiyahManton1989} and uses self-dual gauge fields on $\R^4$ (instantons). In brief, one obtains a degree $N$ Skyrme field as the holonomy of an $N$-instanton along all lines parallel to a given direction. Recently \cite{CorkHarlandWinyard2021gaugedskyrmelowbinding,corkhalcrow2022adhm,harland2023approximating} it was shown that these holonomies may be approximated directly using the Atiyah--Drinfeld--Hitchin--Manin (ADHM) description \cite{ADHM1978construction} of the instanton moduli space. These reproduce accurate approximations of minimal energy skyrmions \cite{corkhalcrow2022adhm} with energies within $2\%$ of numerical minimisers, and allow for the construction of configuration spaces useful for the purpose of quantisation \cite{corkhalcrow2024quantization}.

Typically a numerical approximation is required in order to obtain a Skyrme field from instanton holonomies. This may be done by discretising the line along which the holonomy is taken and approximating the parallel transport along each segment. In \cite{harlandsutcliffe2023rational}, Harland--Sutcliffe used an ultra-coarse discretisation of four intervals $\R=(-\infty,-\mu]\cup[-\mu,0]\cup[0,\mu]\cup[\mu,\infty)$ to generate explicit Skyrme fields from ADHM data, dependent on an auxiliary parameter $\mu>0$, with this value being chosen to find the minimal-energy approximation. Somewhat surprisingly, this gives a better approximation of the minimal-energy $N=1$ solution than taking the full holonomy. Recently in \cite{sutcliffe2024jnr}, Sutcliffe showed how to generate these explicit Skyrme fields from the class of JNR instantons \cite{Jackiw1977}, and showed they too provide good approximations of true skyrmions, with energies only slightly higher than those of the full holonomy when $N>1$. It is this class of Skyrme fields that are the main focus of this paper.

For $N+1$ fixed \textit{JNR poles} $a_i\in\R^4$ and \textit{weights} $\lambda_i>0$, for $i=0,\dots,N$ and auxiliary parameter $\mu>0$, define
\begin{align}\label{JNR-sky}
    U=\frac{\psi(\mu)\ol{\psi(-\mu)}}{|\psi(\mu)\ol{\psi(-\mu)}|},
\end{align}
where $\psi(\mu)=|\zeta|^2\rho(\mu)+\iota(\mu)$ with
\begin{align}\label{JNR-sky-functions}
    \begin{aligned}
        \zeta&=\sum_{i=0}^N\frac{\lambda_i^2(x-a_i)}{|x-a_i|^2},\quad \rho(\mu)=\sum_{i=0}^N\frac{\lambda_i^2}{|x-a_i+\mu|^2},\\
        \iota(\mu)&=\sum_{i,j,k=0}^N\frac{\mu\lambda_i^2\lambda_j^2\lambda_k^2(a_i-a_j)(x-a_i)\ol{(a_i-a_k)}}{|x-a_i+\mu|^2|x-a_j+\mu|^2|x-a_i|^2|x-a_k|^2},
    \end{aligned}
\end{align}
where now $x=\vec{x}$ is a pure imaginary quaternion representing a point in $\R^3$. So long as $\mu>0$, \eqref{JNR-sky} defines a degree $N$ Skyrme field \cite{sutcliffe2024jnr,harlandsutcliffe2023rational} in terms of the free data of JNR poles and weights.

One can show \cite{sutcliffe2024jnr} that $\iota(\mu)$ is a pure imaginary quaternion. In particular, when all of the poles are pure imaginary quaternions, it thus follows that $\psi(-\mu)=\ol{\psi(\mu)}$ and so \eqref{JNR-sky} is a rational function $U = [{\psi(\mu) / \abs{\psi(\mu)}}]^2$.

    It is important to note that the formulation here differs slightly to that of \cite{sutcliffe2024jnr}. There the Skyrme field \eqref{JNR-sky} is instead given by $U=\frac{\psi(-\mu)\ol{\psi(\mu)}}{|\psi(-\mu)\ol{\psi(\mu)}|}$, which is the inverse of our field \eqref{JNR-sky}. The reason for our choice of convention is that in the derivation in \cite{sutcliffe2024jnr} the ADHM construction for \textit{anti-self-dual} JNR instantons is used, which leads to Skyrme fields with negative degree. This adjustment means we are studying skyrmions with positive degree given by $N>0$. This discrepancy did not matter for the results of \cite{sutcliffe2024jnr} as the energy is invariant under $U\mapsto U^{-1}$, however this is important for us since, as we shall discuss in the next sections, we will be interested in the local degree of these fields for which the sign plays a crucial role.
\section{Skyrmion locations and local degree}\label{sec: locations}
For a Skyrme field $U:\R^3\lto\SU(2)$, the \textit{anti-vacua} are the points for which $U=-1$. For the $N=1$ skyrmion, which is spherically symmetric, there is one anti-vacuum at the centre. Furthermore, for a product of $N$ well-separated unit skyrmions, the individual skyrmions are seen to be roughly positioned at the anti-vacua. With this in mind, we will think of the points where $U=-1$ as the \textit{locations} of constituent $\pm1$-skyrmions which form a generic $N$-skyrmion.
\begin{remark}
    This definition of location in terms of anti-vacua is also natural from an energetic perspective when one considers the \textit{pion mass potential} $m_\pi\tr(1-U)$, which is maximal at $U=-1$.
\end{remark}
The locations need to be counted with signed multiplicity due to the fact that we know that $N$ is the topological degree of $U$ as a map $S^3 \lto S^3$, which is the signed count of preimages of a generic point in the target $S^3$, with the sign accounting for whether the map is locally orientation preserving or reversing at the preimage. Another equivalent definition of this sign comes from the \emph{local degree} \cite[p.~136]{Hatcher2002}. These are defined as follows: let $\lbrace x_i \rbrace = U^{-1}(-1)$, and pick $V_i \ni x_i$ small disjoint open neighbourhoods in $\mathbb{R}^3 \subset S^3$ mapped into $W \ni -1$ a small open neighbourhood in $S^3$. These give rise to homomorphisms $H_3(V_i, V_i \setminus \lbrace x_i \rbrace) \to H_3(W, W\setminus\lbrace-1\rbrace)$ which define the local degree, and this definition is independent of the choice of $V_i$. It is a known result that the sum of the local degrees of $U$ at the $x_i$ is equal to the degree of $U$ \cite[Proposition 2.30]{Hatcher2002}. 

Choosing $V_i$ to be sufficiently small balls centred at the $x_i$ and denoting $S_i := \partial V_i$, then by the definition of relative homology and the fact $V_i\setminus\lbrace x_i\rbrace$ deformation retracts onto $S_i$ the local degree is given by the homomorphism $H_2(S_i) \to H_2(\partial W)$. This we shall denote with $\deg(\left. U \right \rvert_{S_i})$. In order to make progress with this definition of local degree, at this point we shall restrict attention to Skyrme fields $U$ which take the form $U=\left(\frac{q}{\abs{q}}\right)^2$ where $q=q_s+q_v$ is a quaternion with real and imaginary parts $q_s$ and $q_v$ respectively. We use the freedom in the choice of sign of $q$ to take $q_s \geq 0$ everywhere. 
This naturally includes several explicit families of Skyrme fields: rational fields arising via the Harland--Sutcliffe approach from $S^1$-invariant ADHM data \cite{harlandsutcliffe2023rational}, and JNR Skyrme fields where the poles have been taken to be pure imaginary quaternions. For Skyrme fields of this form we see that knowledge of the imaginary part $q_v$ is sufficient to calculate the local degree. 

\begin{lemma}\label{lem:local-degree}
Let $S_i$ be the sphere centred at $x_i \in U^{-1}(-1)$ of `sufficiently small' radius $\epsilon$. The local degree of $U$ at $x_i$, $\deg(\left. U \right \rvert_{S_i})$, is the degree of $\left. \wh{q_v}\right\rvert_{S_i} : S^2 \to S^2$. 
\end{lemma}
\begin{proof}
    We have 
    \begin{align*}
U = \frac{(q_s^2 - \abs{q_v}^2) + 2 q_s q_v}{q_s^2 + \abs{q_v}^2} = \frac{(-1 +\delta^2) + 2\delta \wh{q_v}}{1 + \delta^2},
\end{align*}
where $\delta = q_s / \abs{q_v} \geq 0$. For any sufficiently small $\epsilon$ we have $\min_{S_i} \delta >0$, and moreover as $\epsilon \to 0$ we have $\max_{S_i}\delta := \delta_m \to 0$. For sufficiently small $\epsilon$ we thus have that $\left. U\right\rvert_{S_i}$ is homotopic to $-1+2\delta_m \left.\wh{q_v}\right\rvert_{S_i}$, and the result follows.
\end{proof}

\subsection{JNR skyrmion locations}
We shall now restrict attention to JNR Skyrme fields \eqref{JNR-sky} and shall only consider JNR poles $a_i$ which are pure imaginary: this is so that $U=[\psi(\mu)/|\psi(\mu)|]^2$ aligning us with the class of Skyrme fields discussed above. For such fields, the points at which $U=-1$ are given by when $\psi(\mu)$ is pure-imaginary, which corresponds with the zeros of the rational function
\begin{align}
    \zeta(x)=\sum_{i=0}^N\frac{\lambda_i^2(x-a_i)}{|x-a_i|^2}.
\end{align}
Before proceeding further we will state some facts about these locations that Sutcliffe identifies \cite{sutcliffe2024jnr}.
\begin{itemize}
    \item In the limit $|a_0|=\lambda_0\lto\infty$ one obtains the 't Hooft ansatz \cite{CorriganFairlie1977} for which the locations are given precisely by the remaining poles.
    \item If one weight $\lambda_j \to \infty$, then the locations tend to $\lbrace{a_i \, | \, i \neq j}\rbrace$.
    \item The locations are independent of $\mu$.
\end{itemize}
One statement made in \cite{sutcliffe2024jnr} is that, due to the topological degree being $N$, the number of locations cannot exceed $N$. As we shall discuss below, and see via examples later, this statement is not true in general for JNR Skyrme fields; on the other hand this statement is true for Skyrme fields generated via the rational map ansatz \cite{HoughtonMantonSutcliffe1998rational} as rational maps are holomorphic and thus orientation preserving. In private correspondence with Sutcliffe he has indicated that the intended meaning to convey was that generically\footnote{In order to allow $N+1$ locations, there must be negatively signed locations and at least one location of multiplicity $>1$. However, based on observations through several examples, the only situations we expect to allow higher multiplicity are when the poles are planar, where as we show in Section \ref{sec: N=2} (see Remark \ref{remark: N planar poles}) negatively signed locations are prohibited.} the number of locations can never equal $N+1$, i.e. the poles are not to be viewed as the locations.

We can moreover make some generic statements.
\begin{proposition}\label{prop.locations-convex-hull}
    The locations of the skyrmions are in the convex hull of the poles. 
\end{proposition}
\begin{proof}
    This follows immediately from the Generalised Lucas Theorem \cite{Diaz1977}, but it is also instructive for later purposes to see by direct calculation. Indeed, notice that for $x\in\R^3\setminus\{a_0,\dots,a_N\}$, we have $\zeta(x)=0$ if and only if
    \begin{align}\label{crit-point-formula}
        x=\sum_{i=0}^N\left(\frac{\lambda_i^2}{|x-a_i|^2\rho(0)}\right)a_i,
    \end{align}
    where $\rho$ is as in \eqref{JNR-sky-functions}. This is in the convex hull of $a_i$ as $\frac{\lambda_i^2}{|x-a_i|^2\rho(0)}>0$, and
    \begin{align*}
        \sum_{i=0}^N\left(\frac{\lambda_i^2}{|x-a_i|^2\rho(0)}\right)=\frac{\sum_{i=0}^N\frac{\lambda_i^2}{|x-a_i|^2}}{\sum_{i=0}^N\frac{\lambda_i^2}{|x-a_i|^2}}=1.
    \end{align*}
\end{proof}

To analyse the zeros of $\zeta$ in more detail, it is useful to observe
\begin{align}\label{crit-f-zeta}
\zeta = \nabla f \quad \text{where} \quad f(x) = \sum_i \lambda_i^2 \log(\abs{x-a_i}).
\end{align}
As such we can identify points where $\zeta=0$ with critical points of $f$. In the case of equal weights, critical points of $f$ were studied in detail in \cite{Bauer2017}: they allowed $f:\R^n\setminus\{a_i\}\lto\R$ but we are only interested in the case $n=3$. We summarise some key results relevant in our context here.
\begin{itemize}
    \item \cite[Theorem 5.4]{Bauer2017}. The critical points of $f$ are isolated. Moreover, at every critical point of $f$, the Hessian $Hf$ has positivity index at least $2$; in particular this means away from the poles $a_i$, $f$ has only local minima (negativity index $0$) or saddle points (negativity index $1$).
    \item The above and standard Morse theory applied to the function $F=\prod_i|x-a_i|^{\lambda_i^2}$ gives \cite[Theorem 1.3]{Bauer2017}: when $f=\log F$ is a (local) Morse function, there are exactly $N+2h$ critical points, where $h$ are local minima and $N+h$ are saddle points.
    \item \cite[Proposition 6.1]{Bauer2017}. If the poles $a_i$ are a union of $G$-orbits, where $G$ is a finite subgroup of $\O(3)$ with $3$-dimensional irreducible representation, then $0$ is a critical point of $f$, and it is a local minimum when the $a_i$ are non-planar.
    \item There can be arbitrary local minima \cite[Proposition 9.1]{Bauer2017}. For every $h\geq2$, there exists a choice of $3h$ poles in $\R^3$ such that $f$ has $h$ local minima; the choice constructed assigns the poles to specific points on an equilateral triangular prism.
\end{itemize}

\subsection{Signs of locations}\label{sec-sign-loc}
We now want to link the degree of $U$ and the location of the corresponding skyrmions using the JNR ansatz. Each of these locations is equipped with a signed multiplicity, and this information is encapsulated in the local degree $\deg(\left.U\right|_S)$ where $S$ is a small sphere around a given location. For the JNR Skyrme field \eqref{JNR-sky}-\eqref{JNR-sky-functions} with pure imaginary poles this is the degree $\deg(\left.\wh{\iota}(\mu)\right|_S)$ as a map $S^2\lto S^2$, where $\iota(\mu)$ is the imaginary part of the defining function $\psi(\mu)$ for the rational JNR Skyrme fields \eqref{JNR-sky-functions}.

\begin{example}
    Let's look at this in the simple case of $N=1$ with $a_0 = q = -a_1$ a unit quaternion, and equal weights (set to unity). This gives rise to a single location at $x=0$ (see \cite{sutcliffe2024jnr} and \S\ref{sec: N=1}). Near $x=0$ we find 
    \begin{align*}
        \iota(\mu) &\sim \frac{1}{\mu^3} \frac{8 qx\overline{q}}{\abs{x-q}^2 \abs{x+q}^2} .
    \end{align*}
    This conjugation by $q$ corresponds to a rotation of $x$, and so clearly the map $x \mapsto \iota(\mu)$ is orientation-preserving, and thus the local degree is $+1$ as expected. 
\end{example}

Recall the identification of locations with critical points of $f(x) = \sum_i \lambda_i^2 \log(\abs{x-a_i})$. The results of \cite{Bauer2017} summarised above, in particular that the total number of critical points is $N+2h$ where $h$ is the number of local minima, suggests a relationship between the signs of the locations and the negativity index of the critical points. This relationship is established by the following main result of this section.

\begin{theorem}\label{thm:signs-locations-via-f}
Suppose $x_\ast$ is a non-degenerate critical point of $f$, and let $S$ be the small sphere centred at $x_\ast$ as in Lemma \ref{lem:local-degree}. Then ${\rm sgn}\deg(\left. U \right \rvert_S) = - \operatorname{sgn}\left(\det(Hf({x_\ast}))\right)$.
\end{theorem}
In order to prove this, we will require the following, somewhat technical, Lemma.
\begin{lemma}\label{lem:sign-sig}
    Let $N\geq1$, $\{a_0,a_1,\dots,a_N\}\in\R^3$ be distinct, $\lambda_i>0$ and $\mu>0$. Let $x_\ast\in\R^3\setminus\{a_i\}$ be such that $\zeta(x_\ast)=0$. Then
    \begin{align}
    S&=\sum_{i,j = 0}^N \frac{\lambda_i^2 \lambda_j^2 \ip{a_j - a_i,a_j-x_\ast}}{(\abs{a_i-x_\ast}^2 + \mu^2) (\abs{a_j-x_\ast}^2 + \mu^2) \abs{a_j-x_\ast}^2}>0.
\end{align}
\end{lemma}
\begin{proof}
    See Appendix \ref{app:conj-sign}
\end{proof}
We now prove the main result
\begin{proof}[Proof of Theorem \ref{thm:signs-locations-via-f}]
We shall show that $\operatorname{sgn}\det(Hf({x_\ast})) =-\operatorname{sgn}\det(D_{x_\ast} \iota)$. This gives us the sign of the local degree ${\rm sgn}\deg(\left. U \right \rvert_S)$ for non-degenerate points as by Lemma \ref{lem:local-degree} this is determined by whether $\hat{\iota}$ is locally orientation preserving / reversing at $x_\ast$. 

In the derivation by Sutcliffe \cite{sutcliffe2024jnr}, we see that $\iota(\mu)=\Im(\Psi(\mu)\ol{\zeta})$, where
\begin{align*}
    \Psi(\mu) &= \sum_{i,j = 0}^N \frac{\lambda_i^2 \lambda_j^2 (x - a_i + \mu)\overline{(x - a_j + \mu)}(x - a_j)}{\abs{x  -a_i + \mu}^2 \abs{x - a_j + \mu}^2 \abs{x - a_j}^2}\\
    &=\rho\zeta+\sum_{i,j = 0}^N \frac{\lambda_i^2 \lambda_j^2 (a_j - a_i)\overline{(x - a_j + \mu)}(x - a_j)}{\abs{x  -a_i + \mu}^2 \abs{x - a_j + \mu}^2 \abs{x - a_j}^2}\\
    &=\rho\zeta+\sum_{i,j = 0}^N \frac{\mu\lambda_i^2 \lambda_j^2 (a_j - a_i)(x - a_j)}{\abs{x  -a_i + \mu}^2 \abs{x - a_j + \mu}^2 \abs{x - a_j}^2}=:\rho\zeta+\Xi(\mu);
\end{align*}
here in the last line we expanded out the middle bracket and observed the first term vanishes by anti-symmetry in $i\leftrightarrow j$. In the case where the poles are pure imaginary, then $\zeta$ is pure imaginary. Let's think of $\iota$, $\Im(\Xi)$, and $\zeta$ as $3$-vectors. Then
\begin{align}
\iota=-\Re(\Xi)\zeta-\Im(\Xi)\times\zeta,\quad|\iota|^2=|\zeta|^2|\Xi|^2.
\end{align}
For $x_\ast\in\zeta^{-1}(0)$, since $\zeta=\nabla f$, we obtain the differential $D_{x_\ast}\iota:T_{x_\ast}\R^3\lto T_{\iota(x_\ast)}\R^3$ as
\begin{align*}
    D_{x_\ast}\iota=-Hf(x_\ast)\left(\Re(\Xi)\id_3+X\right)|_{x=x_\ast},
\end{align*}
where $X$ is the $3\times 3$ anti-symmetric matrix
\begin{align*}
    X=\begin{pmatrix}
        0&\Im(\Xi)_3&-\Im(\Xi)_2\\
        -\Im(\Xi)_3&0&\Im(\Xi)_1\\
        \Im(\Xi)_2&-\Im(\Xi)_1&0
    \end{pmatrix}.
\end{align*}
As such
\begin{align}\label{detDiota}
    \det(D_{x_\ast}\iota)=-\det(Hf(x_\ast))\left(\Re(\Xi)\left(\Re(\Xi)^2+|\Im(\Xi)|^2\right)\right)|_{x=x_\ast}.
\end{align}
Now, again viewing $a_i$ and $x_\ast$ as $3$-vectors, we have
\begin{align}
    \Re(\Xi)|_{x=x_\ast}&=\sum_{i,j = 0}^N \frac{\mu\lambda_i^2 \lambda_j^2 \ip{a_j - a_i,a_j-x_\ast}}{(\abs{x_\ast  -a_i}^2 + \mu^2)(\abs{x_\ast - a_j}^2 + \mu^2) \abs{x_\ast - a_j}^2}=\mu S(x_\ast).\label{Re(Xi)}
\end{align}
This is always positive by Lemma \ref{lem:sign-sig} and since $\mu>0$. The result thus follows by considering the sign of the formula \eqref{detDiota}.
\end{proof}

As a result of Theorem \ref{thm:signs-locations-via-f} we can thus use the existing literature about the critical points of $f$ as discussed in the previous section. In particular see that (non-degenerate) local minima give $-1$ sign and saddle points give $+1$. These are the only two configurations that arise in the case of equal weights \cite{Bauer2017}. Also by \cite[Proposition 6.1]{Bauer2017}, the JNR skyrmions of equal weights where the poles have the symmetry group of a regular solid (tetrahedral, octahedral, or icosahedral) must admit a negatively signed location at the origin. In the most symmetric cases, we can generate an $N=3$ tetrahedron, $N=5$ octahedron, $N=7$ cube, $N=11$ icosahedron, and an $N=19$ dodecahedron, using the JNR ansatz by placing the poles at the vertices of these Platonic solids, and these are all guaranteed to have a negatively signed location at the origin. The first approximates a minimal-energy skyrmion (which we discuss in detail in \S\ref{sec: tetrahedral 3-skyrmion}), whereas the latter are expected to be saddle points of the Skyrme energy. This result supports early analysis of Houghton--Krusch in \cite{houghtonkrusch2001folding}, where they predicted existence of negative baryon density in some of these examples using non-holomorphic rational maps.

\section{Examples}\label{sec: examples}

\subsection{\texorpdfstring{$N=1$}{N=1}}\label{sec: N=1}

We shall not say much about this, only that this case was considered by Sutcliffe when the weights were equal, and that indeed by an overall translation we can choose $a_0 = q = -a_1$ for some quaternion $q$, and then there is a single location at $x=0$ with multiplicity $+1$. 

\subsection{\texorpdfstring{$N=2$}{N=2}}\label{sec: N=2}

In the case of $N=2$ it is known that the JNR data describes the full $N=2$ instanton moduli space. In fact, the weights and poles do not uniquely determine an instanton: gauge-equivalent JNR are related by Poncelet's porism. In particular, let $T$ be the triangle with vertices $a_i$, $C$ the circumcircle of $T$, and $D$ the inellipse tangent at points $b_i$ along the edges of $T$ such that 
\[
\frac{\lambda_0^2}{\lambda_1^2} = \frac{\abs{a_0 - b_2}}{\abs{b_2 - a_1}}, \quad \text{etc.}
\]
Poncelet's porism identifies $T$ as one triangle in a 1-parameter family of triangles who are circumscribed by $C$ and have $D$ inscribed, and all such triangles give rise to gauge-equivalent JNR data \cite{hartshorne1978stable}. Gauge equivalent instantons in turn give rise to skyrmions equivalent modulo isorotations \cite{AtiyahManton1993}.\footnote{Note the difference between our notation and that of Atiyah \& Manton for the weights: where we use $\lambda_i^2$, they use $\lambda_i$.} Restricting attention again to pure imaginary poles, we can identify the locations of the $N=2$ JNR skyrmions through the geometry of the porism via the following theorem. 

\begin{theorem}\label{thm: 2-skymrion locations are foci}
     The locations of the 2 skyrmions arising from the $N=2$ JNR data $a_0,a_1,a_2\in\R^3$, $\lambda_0,\lambda_1,\lambda_2>0$ are the two foci of $D$, counted with multiplicity $+1$. 
\end{theorem}
Before starting the proof we should note that this result is very natural. As the location is invariant under isorotations, and hence is a gauge-invariant quantity from the JNR perspective, the locations of a skyrmion arising from $N=2$ JNR data must be two points associated to the entire porism, of which the foci are the natural candidates. 
\begin{proof}[Proof of Theorem \ref{thm: 2-skymrion locations are foci}]
    Observe that when we have $N=2$ JNR data there is always a plane through the $a_i$, and so by an overall rotation we can choose these to be of the form $a_i = a_{i1}\bi + a_{i2}\bj+h\bk$. Then writing $x = x_1 \bi + x_2 \bj + x_3 \bk$ we immediately see the $\bk$ component of $\zeta$ is 
    \[
    \sum_{i=0}^N \frac{\lambda_i^2 (x_3-h)}{\abs{x - a_i}^2}, 
    \]
    which is only zero when $x_3=h$. As such when looking for solutions of $\zeta=0$ it suffices to consider the case $h=0$ and restrict to considering $x = x_1 \bi + x_2 \bj $. Let's then factor out an $\bi$ from the $a_i$ and $x$, i.e. write 
    \[
    a_i = \bi(a_{i1} - a_{i2}\bk), \quad x = \bi(x_1 - x_2 \bk).
    \]
    Defining then $z = x_1 + x_2 \bk$, $z_i = a_{i1} + a_{i2}\bk$ we can rewrite 
    \begin{align*}
        \zeta &= \sum_{i=0}^2 \frac{\lambda_i^2 \bi (\overline{z} - \overline{z_i})}{\abs{\bi(\overline{z} - \overline{z_i})}^2} = \bi \sum_{i=0}^2 \frac{\lambda_i^2}{z - z_i}.
    \end{align*}
    The result then reduces down to a theorem of Siebeck (often called Marden's theorem), of which we shall reproduce the statement from \cite{Marden1945}:
    \begin{displayquote}
        The zeros of the partial fraction 
        \[
        Q(z) = \frac{m_0}{z - z_0} + \frac{m_1}{z - z_1} + \frac{m_2}{z - z_2}, \quad m_0 m_1 m_2 \neq 0, 
        \]
        where $z_0$, $z_1$ and $z_2$ are three distinct, non-collinear points, lie at the foci of the conic which touches the line segments $(z_1, z_2)$, $(z_2, z_0)$ and $(z_0, z_1)$ in the points $\zeta_0$, $\zeta_1$ and $\zeta_2$ that divide these segments in the ratio $m_1 : m_2$, $m_2:m_0$, and $m_0:m_1$ respectively.  
    \end{displayquote}
    This clearly applies with $\zeta = Q(z)$ and $m_i = \lambda_i^2$ to give the desired result.
\end{proof}

\begin{example}[Equal weights]
    When all the $\lambda_i$ are equal we see that the tangency points of $D$ must be the midpoints of the edges of $T$. Such an inellipse has a special name, namely the \textbf{Steiner inellipse}. In this case it is also particularly easy to work out the locations of the foci. Observe that writing $m_0 = m_1 = m_2 = m$, then $Q(z) = m P^\prime(z)/P(z)$, where $P(z) = \prod_i (z - z_i)$, and as such the foci are given by the zeros of $P^\prime(z)$.   
\end{example}

\begin{remark}\label{remark: N planar poles}
    Theorem \ref{thm: 2-skymrion locations are foci} generalises to the situation of $N+1$ planar JNR poles of which no 3 are collinear to say that the locations are the foci (appropriately defined) of the corresponding degree-$N$ plane curve which is tangent to each of the $\frac{1}{2}N(N+1)$ line segments in the ratios $\lambda_i^2/\lambda_j^2$ ($i \neq j$). In the case that these $N$ points form the vertices of an equilateral $N$-gon and all the weights are equal the foci all coincide at the center of the $N$-gon. It is well-known that these configurations form tori, which have an axial symmetry, and so a single location is natural.
    
    The relation of Marden's theorem to the solution of such criticality problems has been previously noted in other contexts \cite{Kaiser1993, Aref1998}. 
\end{remark}

\begin{example}[Degenerate single location]\label{ex: degenerate single location}
When the inellipse $D$ is a circle, by Theorem \ref{thm: 2-skymrion locations are foci}, there is a single location at the centre of this circle. This location must be counted with multiplicity $+2$ as the total degree is $N=2$. In this situation the centre of $D$ lies on the intersections of the angle bisectors, and so 
    \[
    \abs{a_0 - b_2} = \abs{a_0 - b_1}, \quad \text{etc.}
    \]
    From this we determine the conditions 
    \[
\frac{\lambda_0^2 + \lambda_1^2}{\lambda_1^2 + \lambda_2^2} = \frac{\abs{a_0 - a_1}}{\abs{a_1 - a_2}}, \quad \text{etc.}
    \]
We can see that in the case of equal weights the corresponding triangle is an equilateral triangle; this case corresponds to the axially-symmetric $2$-skyrmion. But other solutions can occur. To see this we can rewrite these equations as 
\[
\begin{pmatrix} d_0 & (d_0 - d_2) & - d_2 \\ 
-d_0 & d_1 &  (d_1 - d_0)\\ 
(d_2 - d_1) & - d_1& d_2\end{pmatrix} \begin{pmatrix} \lambda_0^2 \\ \lambda_1^2 \\ \lambda _2^2 \end{pmatrix} = 0,
\]
where $d_0 = \abs{a_1 - a_2}$, etc. This 1-dimensional kernel is given by 
\begin{align}\label{weights-incircle}
\begin{pmatrix} \lambda_0^2 \\ \lambda_1^2 \\ \lambda _2^2 \end{pmatrix} \propto \begin{pmatrix}
    -d_0 + d_1 + d_2 \\ d_0 - d_1 + d_2 \\ d_0 + d_1 - d_2
\end{pmatrix}
\end{align}
provided $d_1 + d_2 \neq d_0$ etc., which is ensured by the triangle inequality provided the poles are not collinear. 

Some examples of these type of configurations were considered as centres of scattering paths in \cite{halcrowwinyard2021consistent}: an example is plotted in \cite[Figure 9]{halcrowwinyard2021consistent} which looks like a croissant (an offset ring). There may be interest in studying the moduli space of $2$-JNR skyrmions with a single location in a semi-classical quantisation, as this would provide a non-linear extension of the vibrational modes of the $2$-skyrmion \cite{GudnasonHalcrow2018vibrational}.
\end{example}

\begin{remark}
    In the event that the JNR poles are collinear the conics $C$ and $D$ both degenerate to the line containing the three poles, and then in the limit the foci lie on the line as well. These are then simply worked out by solving a quadratic in terms of a parameter along the line. 
\end{remark}

\subsection{\texorpdfstring{$N=3$}{N=3} tetrahedron}\label{sec: tetrahedral 3-skyrmion}

So far in \S\ref{sec: N=1} and \S\ref{sec: N=2} we have only seen locations with positive multiplicity, and so we seek an example which demonstrates the existence of negatively signed locations, that is the number of preimages of $U=-1$ exceeds $N$. Following the discussion at the end of \S\ref{sec-sign-loc}, we are drawn to consider a tetrahedrally symmetric $N=3$ JNR skyrmion. Other than the existence of such a location, this example is important for two main reasons. 

Firstly, in the context of Euclidean $\SU(2)$ BPS monopoles such a situation where the number of zeros of the Higgs field exceeds the topological charge, and so some must be counted with opposite multiplicity (often called `anti-zeros'), occurs for the charge-3 tetrahedral monopole \cite{sutcliffe1996zeros}. This configuration sees a single negatively signed preimage at the centre of the tetrahedron, and 4 positively signed preimages along the rays of tetrahedron. Given that a partial motivation for this work is studying further the analogy between skyrmions and monopoles, it is sensible to investigate whether the $N=3$ tetrahedrally-symmetric JNR skyrmion exhibits the same behaviour.

Secondly, the minimal-energy $N=3$ skyrmion is tetrahedrally-symmetric \cite{battye2002skyrmions}, and this is relatively well-approximated by the tetrahedral JNR skyrmion \cite{sutcliffe2024jnr}. Furthermore, negative baryon density in a skyrmion with positive baryon number was observed for a tetrahedral $N=3$ skyrmion in \cite{LeeseManton1994stable, Foster2013}, indicative of constituent anti-skyrmions.
 
A charge-$3$ instanton with tetrahedral symmetry may be generated from JNR data with equal weights (fixed as unity), and poles at the vertices of a tetrahedron in $\R^3$. A sensible orientation takes these poles as
\begin{align}
\begin{aligned}
    a_0&=\lambda(\bi+\bj+\bk),& a_1&=\lambda(\bi-\bj-\bk),\\
    a_2&=\lambda(-\bi+\bj-\bk),& a_3&=\lambda(-\bi-\bj+\bk),
\end{aligned}
\end{align}
with $\lambda>0$ a free parameter representing the scale. Writing $x=x_1\bi+x_2\bj+x_3\bk$, here we have 
\begin{align*}
    \zeta &= \frac{(x_1-\lambda)\bi + (x_2-\lambda)\bj + (x_3-\lambda)\bk}{(x_1-\lambda)^2 + (x_2-\lambda)^2 + (x_3-\lambda)^2} + \frac{(x_1-\lambda)\bi + (x_2+\lambda)\bj + (x_3+\lambda)\bk}{(x_1-\lambda)^2 + (x_2+\lambda)^2 + (x_3+\lambda)^2} \\
    &\phantom{=} + \frac{(x_1+\lambda)\bi + (x_2-\lambda)\bj + (x_3+\lambda)\bk}{(x_1+\lambda)^2 + (x_2-\lambda)^2 + (x_3+\lambda)^2} + \frac{(x_1+\lambda)\bi + (x_2+\lambda)\bj + (x_3-\lambda)\bk}{(x_1+\lambda)^2 + (x_2+\lambda)^2 + (x_3-\lambda)^2}.
\end{align*}
Clearly any solutions of $\zeta=0$ will scale with $\lambda$ and will have the same tetrahedral symmetry, so we can look for just one solution in an octant when $\lambda=1$. One finds that the only solutions are when  $x_1=x_2=x_3=\kappa$ and 
    \[
    \kappa(\kappa-1/3) = 0,
    \]
as well as the tetrahedral rotations of these solutions. Using Theorem \ref{thm:signs-locations-via-f}, we see that there are positively signed solutions near the vertices at $(\lambda/3, \lambda/3, \lambda/3)$ and its tetrahedral rotates, and a negatively signed solution at $0$, matching the situation for monopoles \cite{sutcliffe1996zeros}. The side length of this tetrahedron is then $l = \frac{2\sqrt{2}}{3}\lambda$. 

\begin{remark}
    This process of looking for solutions along the line $x_1=x_2=x_3$ is the same as that which was carried out numerically in \cite{HoughtonSutcliffe1996monopole}.
\end{remark}
It is appropriate to make a few comments on how we determined all the solutions. By clearing the denominator of $\zeta$, the coefficients of $\bi$, $\bj$, $\bk$ give three polynomials $f_1, f_2, f_3 \in \mathbb{R}[x_1,x_2,x_3]$, and we seek to use Gr\"obner bases to find the associated variety. Unfortunately as polynomials in $\mathbb{C}[x_1,x_2,x_3]$ the associated variety has dimension 1, coming from constituent polynomials which are sums of squares, and so have complex roots but not real roots. To find the locations we must exclude the cases of these positive dimensional loci, and just calculate real solutions. When we can write $f_1, f_2, f_3 \in \mathbb{k}[x_1,x_2,x_3]$ for some exact field $\mathbb{k}$, by using exact algebraic methods this procedure is guaranteed to find the locations as values algebraic over $\mathbb{k}$. 

\subsection{\texorpdfstring{$N=3$}{N=3} twisted scattering}\label{subsec: twisted line scattering}

An important family of skyrmions considered in the quantisation of the $3$-skyrmion \cite{Walet:1996he} is that of \emph{twisted line scattering}. Here we take the JNR poles to be 
\begin{align}
\begin{aligned}
    a_0&=\lambda(\bi+\bj+\gamma\bk),& a_1&=\lambda(\bi-\bj-\gamma\bk),\\
    a_2&=\lambda(-\bi+\bj-\gamma\bk),&a_3&=\lambda(-\bi-\bj+\gamma\bk),
\end{aligned}
\end{align}
for $\gamma \in \mathbb{R}$, and $\lambda>0$ an overall scale. For the purpose of studying the locations, the scale is unimportant, and furthermore we only need to consider $\gamma\in[0,\infty)$ as all configurations with $\gamma<0$ are related to those at $-\gamma>0$ by a $\tfrac{\pi}{2}$ rotation and isorotation around the $\bk$ axis. The configuration with $\gamma=0$ is an equilateral square in the plane, which we have already seen in Remark \ref{remark: N planar poles} corresponds to an axially-symmetric skyrmion with a triple location at the origin, and the configuration with $\gamma=1$ is the tetrahedral configuration of \S\ref{sec: tetrahedral 3-skyrmion}. This demonstrates that in this continuously varying configuration there is the possibility of the number of distinct locations changing. 

We investigate this numerically, and there are key regions to investigate. At each we give the number of (possibly repeated) positively and negatively signed locations. The critical values of $\gamma$ when the number of locations changes, so called `splitting points', correspond to when the locations are degenerate critical points of $f$. For example the location at $0$ is degenerate when $\gamma=0$ or $\gamma=\sqrt{2}$.
\begin{enumerate}
    \item ($\gamma=0$, [$+3, -0$], Figure \ref{subfig: g=0}). This is the plane equilateral square, a triply positive location at the origin. 
    \item ($\gamma \in (0, \gamma_0)$, $\gamma_0 \approx 7/8$, [$+4, -1$], Figure \ref{subfig: g=0.5}). There is a single negatively signed location at the origin, and in the $x_1,x_2 \geq 0$ quadrant a positively signed solution along the ray $x_1=x_2$ moving monotonically outwards and upwards. This positively signed location is mirrored in the three other quadrants.  
    \item ($\gamma \in (\gamma_0, \sqrt{2})$, [$+4, -1$], Figure \ref{subfig: g=1.0}). The single negative remains at the origin, but the positively signed locations along the rays have turned around and are moving monotonically inwards while still moving upwards.
    \item ($\gamma \in (\sqrt{2}, \sqrt{\sqrt{2}+1})$, [$+5, -2$], Figures \ref{subfig: g=1.42} and \ref{subfig: g=1.54}). The single negative location at the origin has split into one positive location and two negative locations. The positive location remains at the origin, and the two negative solutions move monotonically along the line $x_1=0=x_2$, one in the $+x_3$ direction and the other in the $-x_3$ direction. The other four positive locations along the rays remains inward and upward moving.
    \item ($\gamma \in (\sqrt{\sqrt{2}+1}, \infty)$, [$+3, -0$], Figure \ref{subfig: g=1.6}). The positive location at the origin remains. The positive locations moving inward in the $x_3 > 0$ region and the negative location moving upwards along the axis collide leaving a single positive location on the axis $x_1=0=x_3$, which moves off upwards. The mirror process happens in the $x_3 < 0$ region.
\end{enumerate}
This is the same schematic behaviour observed in the zeros of the Higgs field in twisted line scattering of Euclidean 3-monopoles \cite[Figure 3]{HoughtonSutcliffe1996monopole}. We visualise this in Figure \ref{fig: twisted line scattering plots}, designed to mirror \cite[Figure 3]{HoughtonSutcliffe1996monopole}.

\begin{figure}
    \centering
         \begin{subfigure}[c]{0.32\textwidth}
         \centering
         \includegraphics[width=\textwidth]{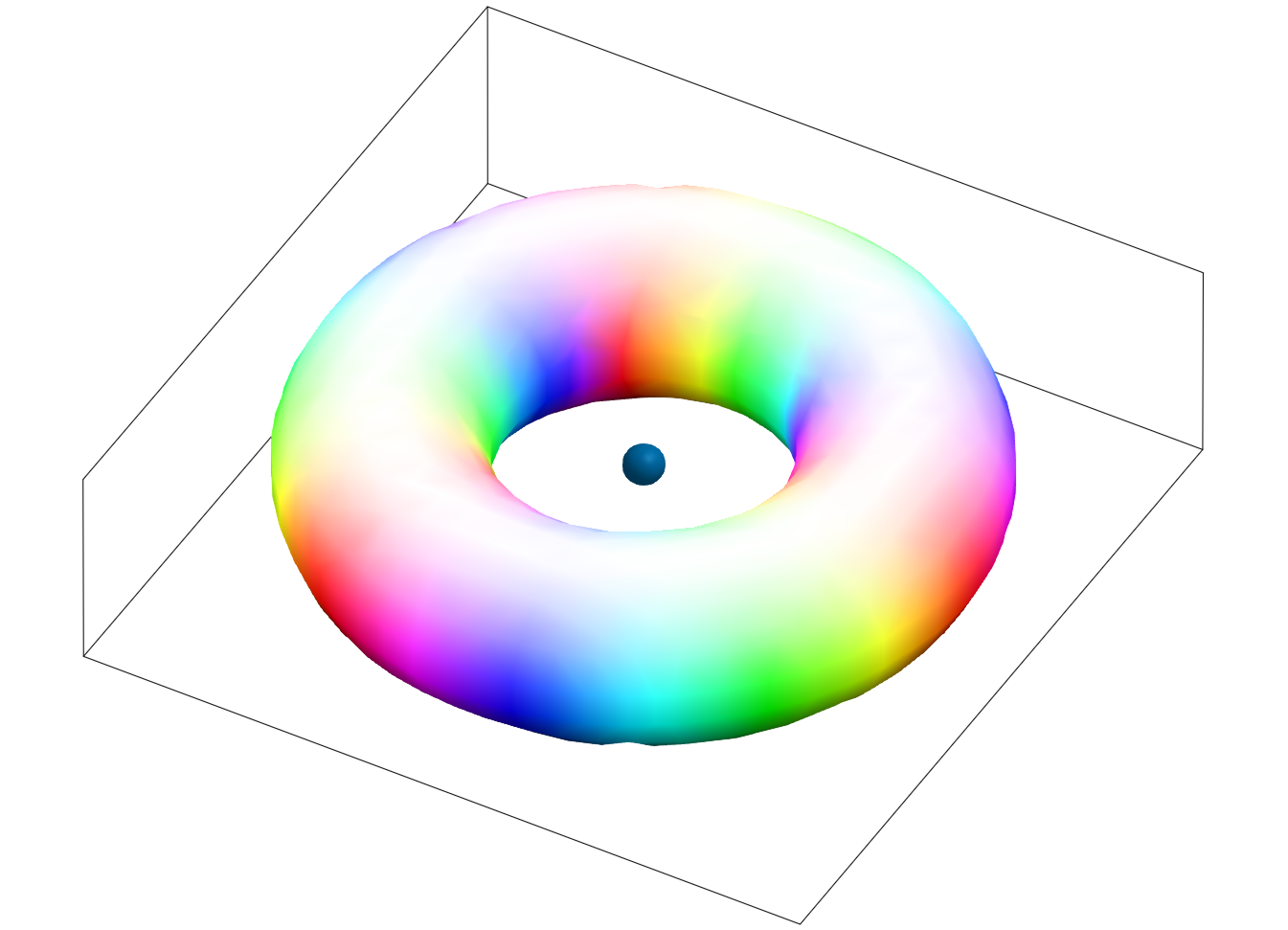}
         \caption{}
         \label{subfig: g=0}
     \end{subfigure}
          \begin{subfigure}[c]{0.32\textwidth}
         \centering
         \includegraphics[width=\textwidth]{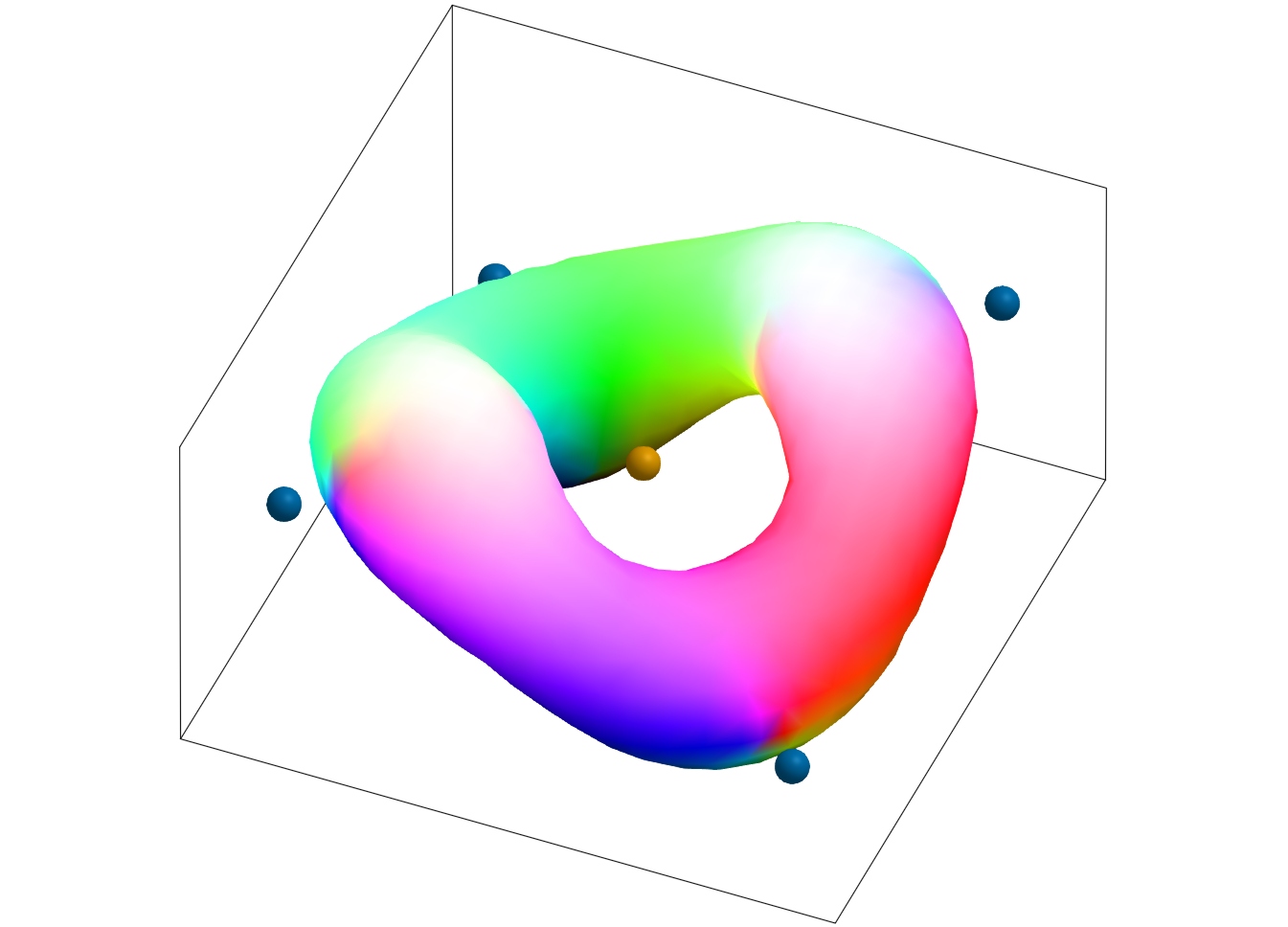}
         \caption{}
         \label{subfig: g=0.5}
     \end{subfigure}
          \begin{subfigure}[c]{0.32\textwidth}
         \centering
         \includegraphics[width=\textwidth]{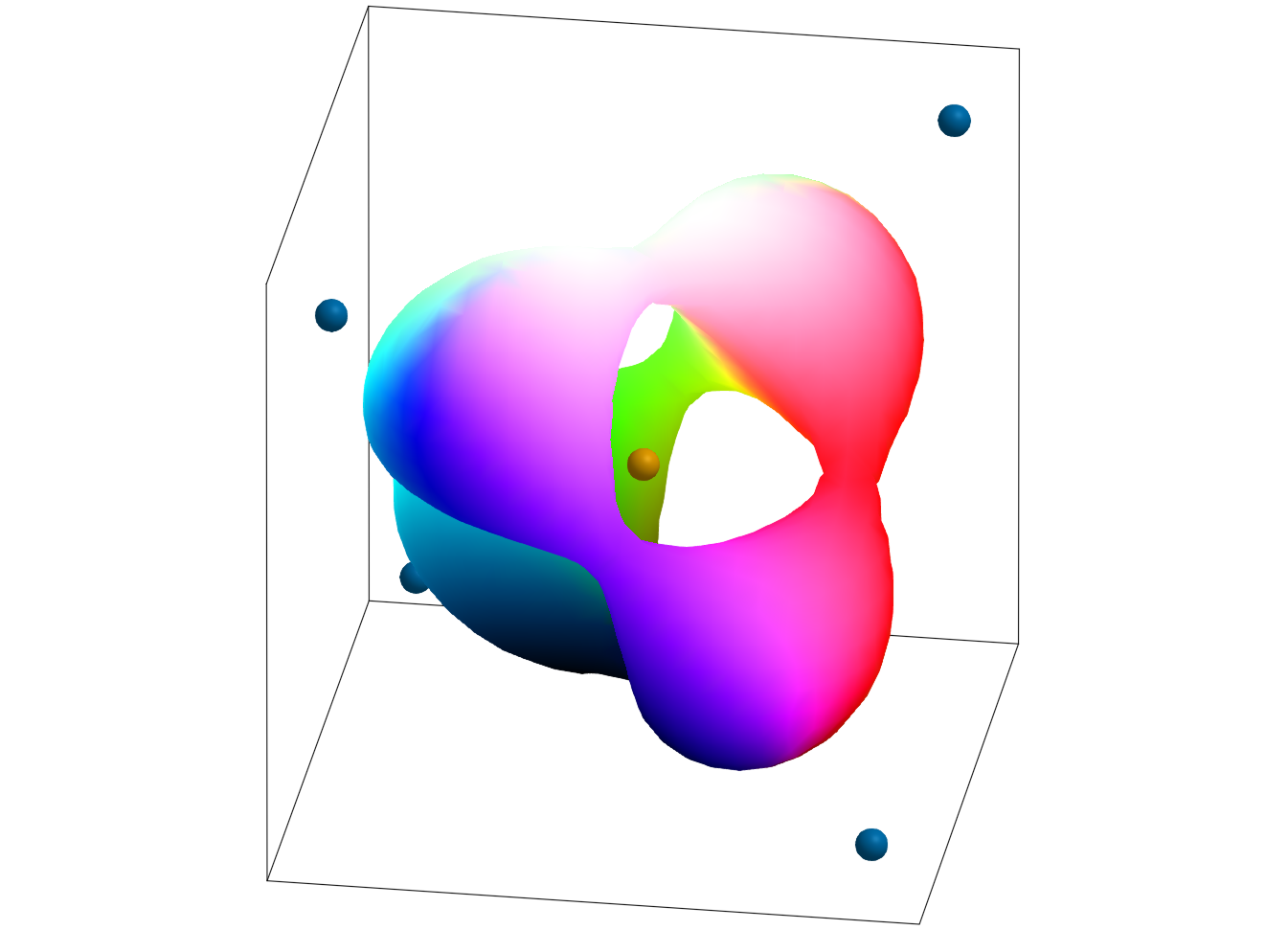}
         \caption{}
         \label{subfig: g=1.0}
     \end{subfigure}
          \begin{subfigure}[c]{0.32\textwidth}
         \centering
         \includegraphics[width=\textwidth]{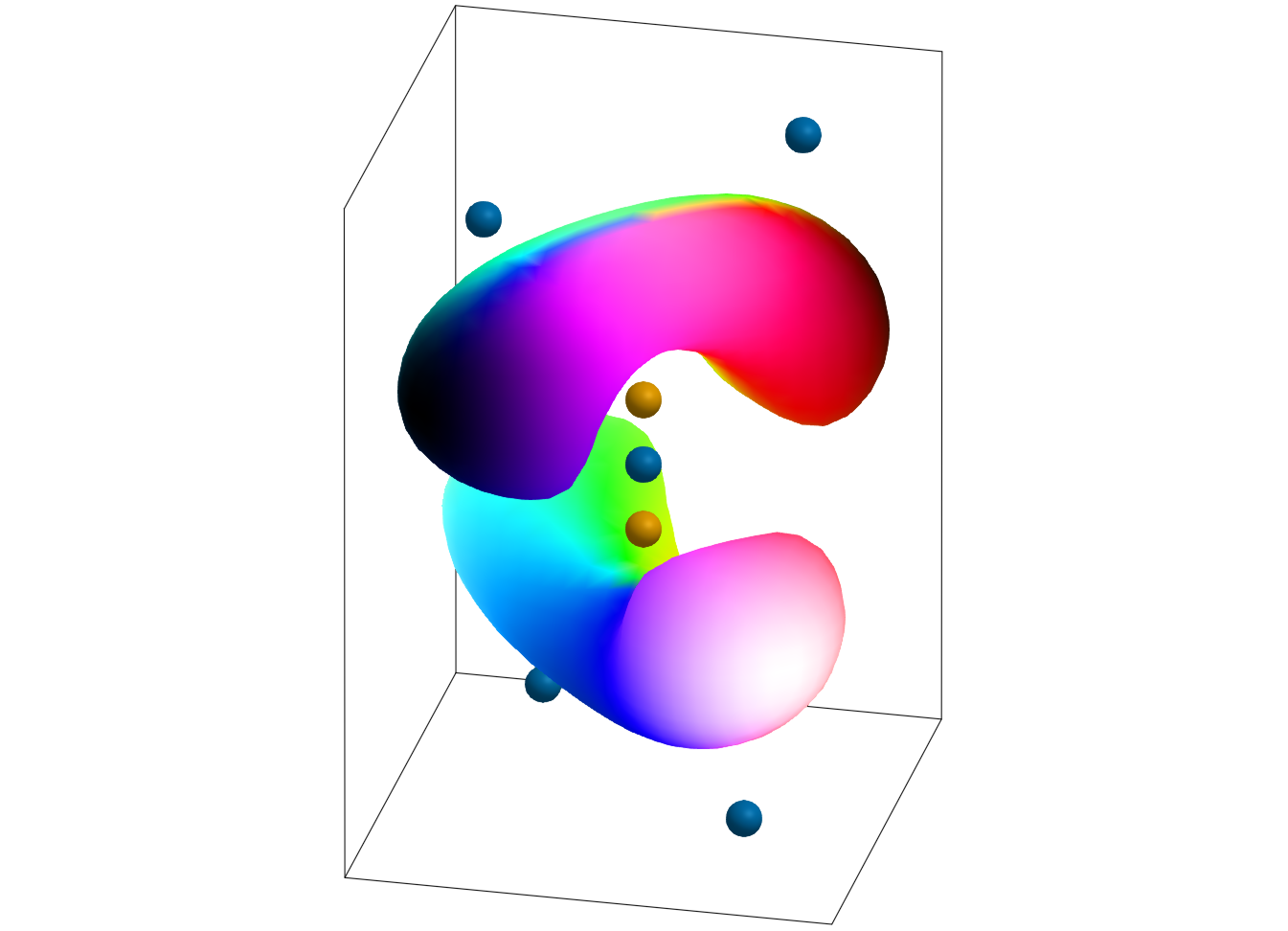}
         \caption{}
         \label{subfig: g=1.42}
     \end{subfigure}
          \begin{subfigure}[c]{0.32\textwidth}
         \centering
         \includegraphics[width=\textwidth]{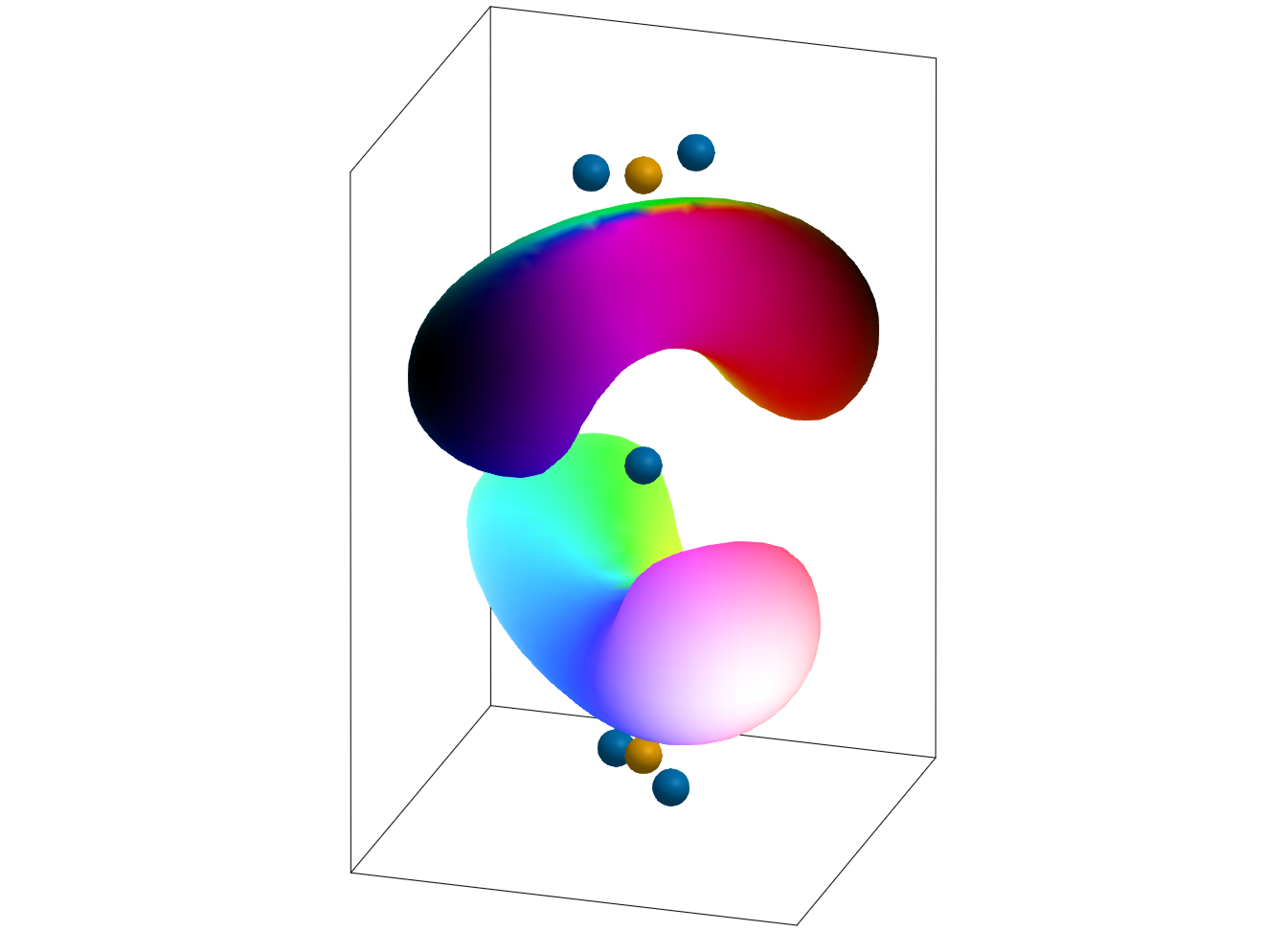}
         \caption{}
         \label{subfig: g=1.54}
     \end{subfigure}
     \begin{subfigure}[c]{0.32\textwidth}
         \centering
         \includegraphics[width=\textwidth]{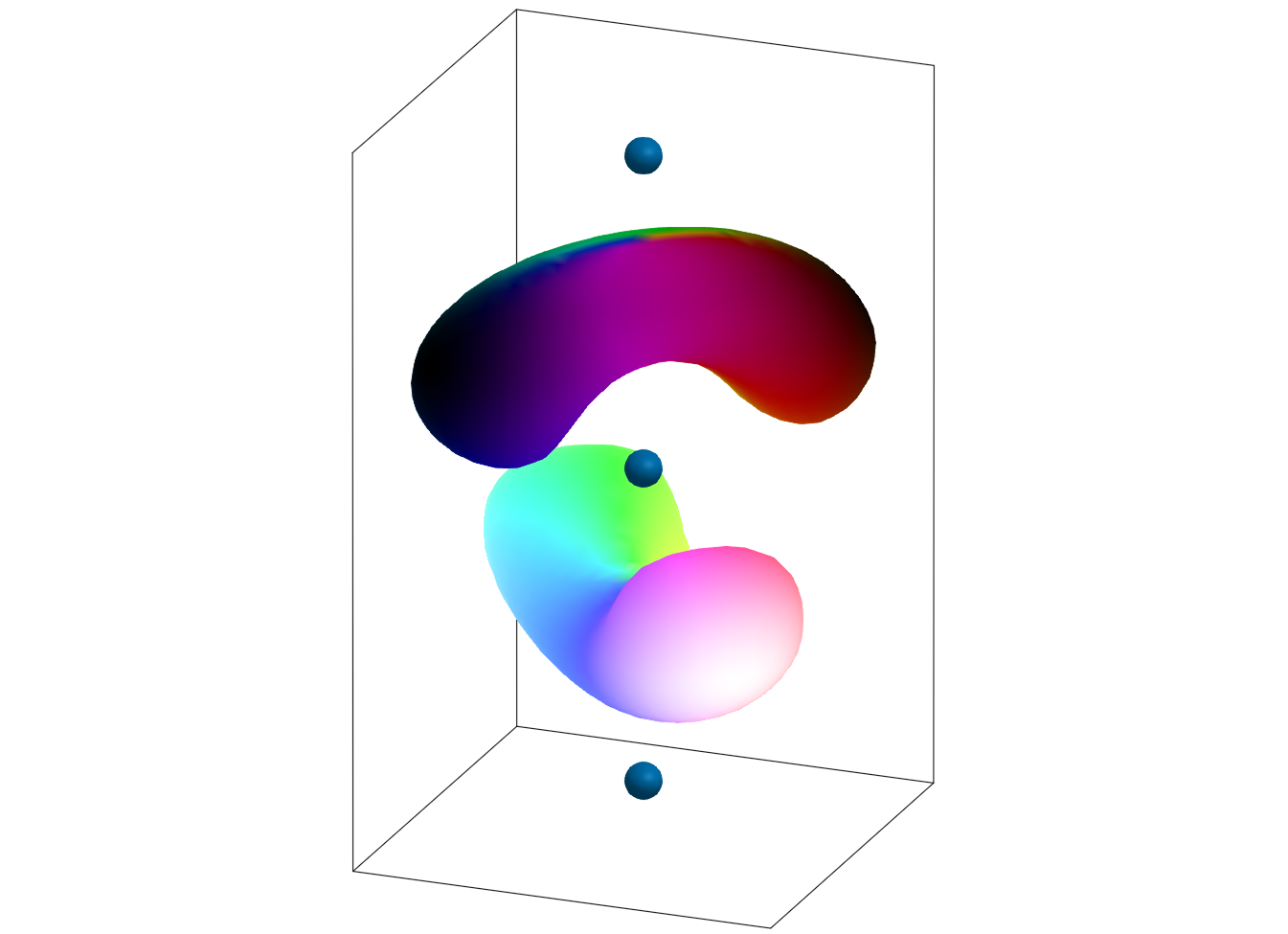}
         \caption{}
         \label{subfig: g=1.6}
     \end{subfigure}
     \caption{Plot of baryon density isosurfaces, and the locations coloured by sign (orange $-1$ and blue $+1$). For a) $\gamma =0$, b) $\gamma = 0.5$, c) $\gamma=1.0$, d) $\gamma=1.4$, e) $\gamma = 1.54$, f) $\gamma=1.6$. The subfigures are not to scale. } 
     \label{fig: twisted line scattering plots}
\end{figure}

    It is appropriate to comment on how Figure \ref{fig: twisted line scattering plots} was produced. Chris Halcrow has written a Julia package `Skyrmions3D' for numerically studying skyrmions \cite{Halcrow23}, and we have written functions within the framework of that package for handling JNR skyrmions. This code, and a tutorial for its use, are at present publicly available from \url{https://github.com/DisneyHogg/JNR_skyrmions}. The baryon density isovalue in the figure is chosen automatically by Skyrmions3D to be $\frac{1}{4}(\text{min. density} + \text{max. density})$. The value of $\lambda=1.51$ and $\mu=3.24$ were fixed through the figures, and chosen to minimise the energy of the tetrahedral skyrmion according to our numerics. We also note that, although not included in the plots of Figure \ref{fig: twisted line scattering plots}, for suitably large $\gamma\gg1$ the baryon density isosurfaces resemble, as expected, three distinct well-separated spherical surfaces with locations inside the surface.

\section{Conclusions and outlook}\label{sec:concl}

This work has developed new analytic and numerical tools to study JNR skyrmions, demonstrating in the process an unexpected depth to the problem of determining the number and sign of constituent locations. Open problems remain, for instance it is not clear without directly calculating all locations how many one should expect from a given JNR configuration, and at present (aside from situations with symmetry) we lack a method to construct a JNR configuration to achieve a desired arrangement of locations without resorting to the 't Hooft limit. Moreover, additional work is required to understand the nature of negatively signed locations. The analogy with BPS monopoles suggests two important avenues of investigation. 

The creation and annihilation of negatively signed locations through the twisted line scattering studied in \S\ref{subsec: twisted line scattering} is behaviour observed for both Euclidean monopoles and JNR skyrmions which cannot be claimed to follow from a connection with rational maps, as skyrmions generated via the rational map ansatz cannot have negatively signed locations. Hyperbolic monopoles have been constructed via JNR data \cite{bolognesiCockburnSutcliffe2014hyperbolicJNR}, but the configurations used are quite different from those studied in this work. While for JNR skyrmions we provide the interpretation of the splitting points as the vanishing of a hessian at critical points, for Euclidean monopoles these are observed in \cite{sutcliffe1996zeros} to coincide with the vanishing of a discriminant of the spectral curve. Whether these two perspectives can be married up warrants additional study.

An observation in \cite{sutcliffe1996zeros} borne out through analysis of Taubes is that for Euclidean monopoles, though one can get negatively signed zeros of the Higgs field inside clusters of positively signed zeros such as with the tetrahedral 3-monopole, asymptotically one can never see well-separated negatively charged clusters in an overall positively charged monopole. An intuition for this arises from \cite{Manton1977}, which shows that there is an asymptotic attractive Coulomb force between an antimonopole and a monopole, meaning that such configurations would not be stable. We might wonder whether the same is true for the JNR skyrmion locations. The locations are the critical points of $f=\sum_i\lambda_i^2\log|x-a_i|$; equivalently these are the critical points of $F=\prod_i|x-a_i|^{\lambda_i^2}$ away from the poles. The poles are the absolute minima of $F$, and by Theorem \ref{thm:signs-locations-via-f} the negatively signed locations are local minima. We have already proven in Proposition \ref{prop.locations-convex-hull} that all locations lie within the convex hull of the poles, and so heuristically one expects any local minima to be somehow shielded by the saddle points (positively signed locations) which would lie somewhere ``between" the local minima and global minima (the poles). However, we have not been able to produce a precise formulation of this intuition.

Given their explicit Skyrme field and simple construction, JNR skyrmions may serve as a useful class of fields with which to begin probing the configuration space of skyrmions. In particular, one may use the Skyrme Lagrangian to get a metric on the moduli space \cite{corkhalcrow2024quantization,AtiyahManton1993}. With our new understanding of skyrmion locations, we are able to give a clear and algebraic definition of `separation' of two constituent 1-skyrmions in an $N=2$ skyrmion, which has proven a challenge in the past \cite{Verbaarschot1987, HosakaOkaAmado1991}. In addition, as highlighted in Example \ref{ex: degenerate single location} our approach lets us identify the family of JNR skyrmions which have a single location with multiplicity $+2$ at the origin, which gives a family whose energy variation must thus be attributed to vibrational modes. 

It is well-known that JNR data does not generate all instantons beyond $N=2$, and in particular many important Skyrme field configurations cannot be generated from the JNR ansatz; the first most important example is the $N=4$ cube \cite{LeeseManton1994stable}. Fortunately, some of the ideas and results presented here are not limited to this class of Skyrme fields. In particular, the result of Lemma \ref{lem:local-degree} applies to any Skyrme field of the form $U=(\frac{q}{|q|})^2$; in this situation the locations are given by $\Re(q)=0$, and the local multiplicity by the degree of $\wh{\Im(q)}$ as a map of $2$-spheres. This motivates revisiting the open problem of \cite{harlandsutcliffe2023rational} of classifying all ADHM data which give rise to rational skyrmions. In fact, given that some important topological information about ADHM skyrmions may be computed from ADHM \cite{corkharland2024FR}, it would be lucrative to also determine a way of computing the locations and signed multiplicity directly from the ADHM data itself.
\section*{Statements and declarations}
\subsection*{Acknowledgements}
JC acknowledges support provided by the University of Leeds School of Mathematics Research Visitor Centre, where much of this work was completed. The research of LDH is supported by a UK Engineering and Physical Sciences Research Council (EPSRC) doctoral prize fellowship.
\subsection*{Data availability statement}
The data that support the findings of this article are available upon reasonable request from the authors. Supporting Julia code is available from \url{https://github.com/DisneyHogg/JNR_skyrmions}.
\subsection*{Conflict of interests}
The authors have no competing interests to declare that are relevant to the content of this article.
\appendix
\section{Proof of Lemma \ref{lem:sign-sig}}\label{app:conj-sign}
Here we prove Lemma \ref{lem:sign-sig}, which states:
\begin{quote}
    Let $N\geq1$, $\{a_0,a_1,\dots,a_N\}\in\R^3$ be distinct, $\lambda_i>0$ and $\mu>0$. Let $x_\ast\in\R^3\setminus\{a_i\}$ be such that $\zeta(x_\ast)=0$. Then
    \begin{align}
    S&=\sum_{i,j = 0}^N \frac{\lambda_i^2 \lambda_j^2 \ip{a_j - a_i,a_j-x_\ast}}{(\abs{a_i-x_\ast}^2 + \mu^2) (\abs{a_j-x_\ast}^2 + \mu^2) \abs{a_j-x_\ast}^2}>0.
\end{align}
\end{quote}
\begin{proof}
By the proof of Proposition \ref{prop.locations-convex-hull}, we see that $\zeta(x_\ast)=0$ if and only if
    \begin{align}\label{zero-convex}
    x_\ast=\frac{\sum_k\tfrac{\lambda_k^2}{|a_k-x_\ast|^2}a_k}{\sum_k\tfrac{\lambda_k^2}{|a_k-x_\ast|^2}}.
\end{align}
Let $v_i=\frac{\lambda_i^2}{|a_i-x_\ast|^2}(a_i-x_\ast)$. Then $v_i\neq0$ and we also have
\begin{align*}
    a_i-x_\ast=\frac{\lambda_i^2}{|v_i|^2}v_i.
\end{align*}
Thus we may write
\begin{align*}
    S&=\sum_{i,j=0}^N\frac{\lambda_i^2\lambda_j^2\ip{\tfrac{\lambda_j^2}{|v_j|^2}v_j-\tfrac{\lambda_i^2}{|v_i|^2}v_i,\tfrac{\lambda_j^2}{|v_j|^2}v_j}}{\left(\tfrac{\lambda_i^4}{|v_i|^2}+\mu^2\right)\left(\tfrac{\lambda_j^4}{|v_j|^2}+\mu^2\right)\tfrac{\lambda_j^4}{|v_j|^2}}=\sum_{i,j=0}^N\frac{\lambda_i^2\lambda_j^2-\ip{v_i,v_j}\tfrac{\lambda_i^4}{|v_i|^2}}{\left(\tfrac{\lambda_i^4}{|v_i|^2}+\mu^2\right)\left(\tfrac{\lambda_j^4}{|v_j|^2}+\mu^2\right)}.
\end{align*}
This is of the form $S=p-q$ where
\begin{align*}
\begin{aligned}
p&=\sum_{i,j=0}^N\frac{\lambda_i^2\lambda_j^2}{\left(\tfrac{\lambda_i^4}{|v_i|^2}+\mu^2\right)\left(\tfrac{\lambda_j^4}{|v_j|^2}+\mu^2\right)}>0,&
q&=\sum_{i,j=0}^N\frac{\ip{v_i,v_j}\tfrac{\lambda_i^4}{|v_i|^2}}{\left(\tfrac{\lambda_i^4}{|v_i|^2}+\mu^2\right)\left(\tfrac{\lambda_j^4}{|v_j|^2}+\mu^2\right)}.
\end{aligned}
\end{align*}
It suffices therefore to prove that $q\leq0$.

Note that
\begin{align*}
    q
    &=\frac{\lambda_0^4}{\left(\tfrac{\lambda_0^4}{|v_0|^2}+\mu^2\right)^2}+\frac{1}{\tfrac{\lambda_0^4}{|v_0|^2}+\mu^2}\sum_{i=1}^N\frac{\left(\tfrac{\lambda_i^4}{|v_i|^2}+\tfrac{\lambda_0^4}{|v_0|^2}\right)\ip{v_i,v_0}}{\tfrac{\lambda_i^4}{|v_i|^2}+\mu^2}\\
    &\qquad\qquad+\sum_{i,j=1}^N\frac{\ip{v_i,v_j}\tfrac{\lambda_i^4}{|v_i|^2}}{\left(\tfrac{\lambda_i^4}{|v_i|^2}+\mu^2\right)\left(\tfrac{\lambda_j^4}{|v_j|^2}+\mu^2\right)},
\end{align*}
and the condition \eqref{zero-convex} gives $\sum_iv_i=0$. As such we can eliminate $v_0$ entirely from this expression by writing
\begin{align}\label{v0-ids}
    v_0&=-\sum_{i=1}^Nv_i,\quad |v_0|^2=\sum_{i,j=1}^N\ip{v_i,v_j}.
\end{align}
Introducing the notation $\lambda_0^4:=\nu$ and $|v_0|^2:=R$ for simplicity of presentation, inserting the identities \eqref{v0-ids} and rearranging, we obtain
\begin{multline*}
    (\nu+R\mu^2)^2q=R^2\nu+\left(\nu+R\mu^2\right)^2\sum_{i,j=1}^N\frac{\tfrac{\lambda_i^4}{|v_i|^2}\ip{v_i,v_j}}{\left(\tfrac{\lambda_i^4}{|v_i|^2}+\mu^2\right)\left(\tfrac{\lambda_j^4}{|v_j|^2}+\mu^2\right)}\\
    -(\nu+R\mu^2)\sum_{i,j=1}^N\frac{\left(\tfrac{\lambda_i^4}{|v_i|^2}R+\nu\right)\ip{v_i,v_j}}{\tfrac{\lambda_i^4}{|v_i|^2}+\mu^2}.
\end{multline*}
This is a polynomial in $\nu$ of the form $(\nu+R\mu^2)^2q=a+b\nu+c\nu^2$, where
\begin{align*}    
a&=-R^2\mu^2\sum_{i,j=1}^N\frac{\tfrac{\lambda_i^4}{|v_i|^2}\tfrac{\lambda_j^4}{|v_j|^2}\ip{v_i,v_j}}{\left(\tfrac{\lambda_i^4}{|v_i|^2}+\mu^2\right)\left(\tfrac{\lambda_j^4}{|v_j|^2}+\mu^2\right)}=-R^2\mu^2\left|\sum_{i=1}^N\frac{\tfrac{\lambda_i^4}{|v_i|^2}}{\tfrac{\lambda_i^4}{|v_i|^2}+\mu^2}v_i\right|^2,\\
b&=R^2+2R\mu^2\sum_{i,j=1}^N\frac{\tfrac{\lambda_i^4}{|v_i|^2}\ip{v_i,v_j}}{\left(\tfrac{\lambda_i^4}{|v_i|^2}+\mu^2\right)\left(\tfrac{\lambda_j^4}{|v_j|^2}+\mu^2\right)}-R\sum_{i,j=1}^N\ip{v_i,v_j}\\
&=2R\mu^2\sum_{i,j=1}^N\frac{\tfrac{\lambda_i^4}{|v_i|^2}\ip{v_i,v_j}}{\left(\tfrac{\lambda_i^4}{|v_i|^2}+\mu^2\right)\left(\tfrac{\lambda_j^4}{|v_j|^2}+\mu^2\right)},\\
c&=\sum_{i,j=1}^N\frac{\left(\tfrac{\lambda_i^4}{|v_i|^2}-\tfrac{\lambda_j^4}{|v_i|^2}\right)\ip{v_i,v_j}}{\left(\tfrac{\lambda_i^4}{|v_i|^2}+\mu^2\right)\left(\tfrac{\lambda_j^4}{|v_j|^2}+\mu^2\right)}-\sum_{i,j=1}^N\frac{\mu^2\ip{v_i,v_j}}{\left(\tfrac{\lambda_i^4}{|v_i|^2}+\mu^2\right)\left(\tfrac{\lambda_j^4}{|v_j|^2}+\mu^2\right)}\\
&=-\mu^2\left|\sum_{i=1}^N\frac{v_i}{\tfrac{\lambda_i^4}{|v_i|^2}+\mu^2}\right|^2;
\end{align*}
in the last line we used the anti-symmetry in the indices $i,j$ to show the first sum is zero. Therefore
\begin{align}
    q=\frac{a+b\nu+c\nu^2}{(\nu+R\mu^2)^2}=-\frac{\mu^2}{(\nu+R\mu^2)^2}\left|\sum_{i=1}^N\frac{\left(R\tfrac{\lambda_i^4}{|v_i|^2}-\nu\right)v_i}{\tfrac{\lambda_i^4}{|v_i|^2}+\mu^2}\right|^2\leq0.
\end{align}
Thus as $p>0$ and $q\leq0$, we have $S=p-q>0$ as required.
\end{proof}
\begin{remark}
    As a consequence of \eqref{zero-convex}, we can write
    \begin{align}\label{sig2}
        S=\frac{1}{\rho(0)}\sum_{i,j,k=0}^N\frac{\lambda_i^2\lambda_j^2\lambda_k^2\ip{y_j-y_i,y_j-y_k}}{(|y_i|^2+\mu^2)(|y_j|^2+\mu^2)|y_j|^2|y_k|^2},
    \end{align}
    where $\rho(\mu)$ is given in \eqref{JNR-sky-functions}, and $y_i=a_i-x_\ast$. The above proof shows that this is positive so long as $\sum_{i}\tfrac{\lambda_i^2}{|y_i|^2}y_i=0$, however we conjecture that the sum in \eqref{sig2} is positive for any choice of $y_i\in\R^3\setminus\{a_i\}$. We have a proof of this for $N=2$, for $N=3$ with equal weights, and when the $y_i$ are on a single sphere, but a general proof without any constraint on the $y_i$ appears to be a hard combinatorial problem.
\end{remark}
\bibliographystyle{unsrt}
\bibliography{refs}
\end{document}